\documentclass[11pt, oneside]{article}
\usepackage{geometry}
\usepackage[utf8]{inputenc} 

\usepackage{geometry}[1in]

\usepackage[T1]{fontenc}    
\usepackage{hyperref}       

\usepackage{url}            
\usepackage{booktabs}       
\usepackage{amsfonts}       
\usepackage{latexsym,amssymb,amsthm,amscd,amsopn,amsmath}
\usepackage{bbm}
\usepackage{mathrsfs}
\usepackage{nicefrac}       
\usepackage{microtype}      
\usepackage{xcolor}         
\usepackage{algorithm}
\usepackage{verbatim}
\usepackage[noend]{algpseudocode}

\usepackage{xspace}
\usepackage{mathtools}
\usepackage{svg}
\usepackage[nameinlink,capitalize]{cleveref}
\usepackage{makecell}

\definecolor{dgreen}{rgb}{0,0.5,0}
\hypersetup{
  colorlinks=true,
  linkcolor=blue,
  filecolor=blue,
  citecolor = dgreen,      
  urlcolor=cyan,
}

\crefformat{equation}{#2(#1)#3}
\Crefformat{equation}{#2(#1)#3}
\Crefname{construction}{Construction}{Constructions}

\Crefformat{figure}{#2Figure #1#3}
\Crefname{assumption}{Assumption}{Assumptions}
\Crefformat{assumption}{#2Assumption #1#3}
\Crefname{subsubsection}{Section}{Sections}
\crefformat{subsubsection}{#2Section #1#3}
\Crefformat{subsubsection}{#2Section #1#3}

\theoremstyle{plain}
\newtheorem{theorem}{Theorem}
\newtheorem{lemma}[theorem]{Lemma}

\theoremstyle{definition}
\newtheorem{definition}[theorem]{Definition}

\numberwithin{theorem}{section}

\newenvironment{namedproof}[1]{\paragraph{Proof of #1.}\hspace{-1em}}{\hfill$\blacksquare$\vspace{1em}}

\newcommand{\nc}{\newcommand}
\nc{\DMO}{\DeclareMathOperator}
\newcount\Comments  
\nc{\Moracle}{\MM^{\mathsf{oracle}}}
\nc{\tilMoracle}{\til{\MM}^{\mathsf{oracle}}}
\nc{\SimulateReduction}{\texttt{SimulateReduction}\xspace}
\nc{\TestReduction}{\texttt{DistinguishReduction}\xspace}
\nc{\SimulateSampling}{\texttt{SimulateSampling}\xspace}
\nc{\SimulateRegression}{\texttt{SimulateRegression}\xspace}
\nc{\Osample}{{\MO_{\mathsf{samp}}}}
\nc{\Oregress}{{\MO_{\mathsf{regress}}}}
\nc{\Oregressp}{{\MO'_{\mathsf{regress}}}}
\nc{\Obandits}{{\MO_{\mathsf{bandits}}}}
\nc{\dom}{\mathsf{dom}}
\nc{\SF}{\mathscr{F}}
\nc{\Fchernoff}{\MF^{\mathsf{chernoff}}}

\DMO{\prox}{prox}
\DMO{\Span}{span}
\DMO{\UCB}{UCB}
\DMO{\LCB}{LCB}
\nc{\expl}[2]{\ME^{#1}_{#2}}
\nc{\tilmdp}[1]{\til \MM({#1})}
\nc{\barpdp}[2]{\ol \MP_{#1}({#2})}
\nc{\barmdp}[1]{\ol \MM({#1})}
\nc{\hatmdp}[1]{\wh \MM({#1})}
\nc{\rem}[2]{\MR_{#1}({#2})}
\nc{\Pigen}{\Pi^{\rm gen}}
\nc{\Pidet}{\Pi^{\rm det}}
\nc{\PiZ}{\Pi_{\SZ}^{\rm markov}}
\nc{\und}[3]{\MU_{{#1}}^{{#2}}({#3})}
\nc{\zlow}[2]{\MZ_{{#1}}^\lowv({#2})}
\nc{\dg}{\dagger}
\nc{\bB}{\mathbf{B}}
\nc{\unif}{\mu_{\rm unif}}
\nc{\indsig}[2]{\mathcal{I}_{#1}({#2})}
\nc{\total}{{\rm fin}}
\nc{\early}{{\rm pre}}
\nc{\zsink}{z_{\rm sink}}
\nc{\lowv}{{\rm low}}
\nc{\oo}[1]{\texttt{o}({#1})}
\nc{\posnrm}[1]{\left[ {#1} \right]_+}
\nc{\negnrm}[1]{\left[ {#1} \right]_-}
\nc{\tvnrm}[1]{\left\| {#1} \right\|_1}
\nc{\absval}[1]{\left| {#1} \right|}
\nc{\normalize}[1]{\mathfrak{n}\left({#1}\right)}

\nc{\SZ}{\textsf{Z}}
\nc{\SO}{\textsf{O}}
\nc{\suff}[2]{{\rm suff}_{#1}({#2})}
\nc{\UPhi}{\mathscr{U}_{X,H}}
\nc{\UPhis}{\til{\mathscr{U}}_{X,H,\MF}}
\nc{\SV}{\mathscr{V}}
\nc{\Phiset}{\Phi_{X,H}}
\nc{\Phisets}{\til{\Phi}_{X,H,\MF}}
\nc{\Lyu}{{\mathtt{Lyu}}}
\nc{\wAlg}{{\widetilde \Alg}}

\nc{\ApproxMDP}{\texttt{ConstructMDP}\xspace}
\nc{\mainalg}{\texttt{BaSeCAMP}\xspace} 
\nc{\bspanner}{\texttt{BarySpannerPolicy}\xspace}

\nc{\gamvec}{\gamma}
\nc{\til}{\widetilde}
\nc{\td}{\tilde}
\nc{\wh}{\widehat}
\nc{\todo}[1]{\ifnum\Comments=1 {\color{red}  [TODO: #1]}\fi}
\nc{\old}[1]{\ifnum\Comments=1 {\color{brown}  [COPIED: #1]}\fi}
\definecolor{darkgreen}{rgb}{0.0, 0.5, 0.0}
\nc{\noah}[1]{\ifnum\Comments=1 {\color{darkgreen} [ng: #1]}\fi}
\nc{\dhruv}[1]{\ifnum\Comments=1 {\color{purple} [dr: #1]}\fi}
\nc{\BP}{\mathbb{P}}
\nc{\BM}{\mathbb{M}}
\nc{\bbapx}{\bb^{\rm apx}}
\nc{\bbapxs}[1]{\bb^{\rm apx, {#1}}}

\nc{\fools}[3]{\MF_{#3}({#1}, {#2})}
\nc{\fool}[2]{\MF({#1},{#2})}
\nc{\clip}[2]{{\rm clip}\left[ \left. {#1} \right| {#2} \right]}
\nc{\imax}{\omega}
\DMO{\conv}{conv}
\nc{\MH}{\mathcal{H}}
\nc{\CH}{\mathscr{H}}
\nc{\CB}{\mathscr{B}}
\nc{\cD}{\mathscr{D}}
\nc{\MC}{\mathcal{C}}
\nc{\st}{\star}
\nc{\lng}{\langle}
\nc{\rng}{\rangle}
\DMO{\OOPT}{opt}
\nc{\dopt}[2]{\ell_{\OOPT}({#1},{#2})}
\nc{\grad}{\nabla}
\nc{\MG}{\mathcal{G}}
\nc{\MP}{\mathcal{P}}
\nc{\PP}{\mathbb{P}}
\nc{\TT}{\mathbb{T}}
\nc{\TTmax}{\TT_{\max}}
\DMO{\Ham}{Ham}
\DMO{\Gap}{Gap}
\DMO{\GD}{GD}
\DMO{\GDA}{GDA}
\DMO{\EG}{EG}
\DMO{\OGDA}{OGDA}
\DMO{\Unif}{Unif}
\DMO{\Tr}{Tr}
\nc{\ul}{\underline}
\nc{\ol}{\overline}
\nc{\Qu}{\ul{Q}}
\nc{\Qo}{\ol{Q}}
\nc{\Ro}{\ol{R}}
\nc{\Vu}{\ul{V}}
\nc{\Vo}{\ol{V}}
\nc{\RanQ}{\Delta Q}
\nc{\RanV}{\Delta V}
\nc{\clipQ}{\Delta \breve{Q}}
\nc{\frzQ}{\Delta \mathring{Q}}
\nc{\clipV}{\Delta \breve{V}}
\nc{\clipdelta}{\breve{\delta}}
\nc{\cliptheta}{\breve{\theta}}
\nc{\delmin}{\Delta_{{\rm min}}}
\nc{\delmins}[1]{\Delta_{{\rm min},{#1}}}
\nc{\gapfinal}[1]{\max \left\{ \frac{\frzQ_{{#1}}^{k^\st}(x,a)}{2H}, \frac{\delmin}{4H} \right\}}
\nc{\post}[2]{R({#1}; {#2})}
\nc{\posts}[3]{R_{#3}({#1}; {#2})}
\nc{\MAJ}{\mathsf{MAJ}}
\nc{\Dnull}{D^{\circ}}

\nc{\BKW}{\mathtt{BKW}}
\nc{\Dec}{\mathtt{Dec}}
\nc{\delreg}{\delta_{\mathsf{reg}}}
\nc{\delreal}{\delta_{\mathsf{real}}}
\nc{\Sreg}{S_{\mathsf{Reg}}}
\nc{\Treg}{T_{\mathsf{Reg}}}
\nc{\PC}{\mathtt{PC}}
\nc{\alphaPC}{\alpha_{\mathsf{PC}}}
\nc{\TPC}{T_{\mathsf{PC}}}
\nc{\SPC}{S_{\mathsf{PC}}}
\nc{\delsmall}{{\delta_{\mathsf{small}}}}
\nc{\CL}{\mathtt{ContrastLearn}}
\nc{\Select}{\mathtt{Select}}
\nc{\piunif}{{\pi_{\mathsf{unif}}}}
\nc{\picov}{{\pi_{\mathsf{cov}}}}
\nc{\ZZ}{\mathbb{Z}}
\nc{\sk}{{\mathsf{sk}}}
\nc{\Enc}{\mathtt{Enc}}
\nc{\EntLPN}{\mathtt{EntangleLPN}}
\nc{\mureg}{{\mu_{\mathsf{reg}}}}
\nc{\PPE}{\mathtt{PPE}}
\nc{\FQI}{\mathtt{FQI}}
\nc{\False}{\mathtt{False}}
\nc{\True}{\mathtt{True}}
\nc{\epreg}{\epsilon_{\mathsf{reg}}}
\nc{\algnst}[1]{\begin{align*}#1\end{align*}}
\nc{\algn}[1]{\begin{align}#1\end{align}}
\nc{\matx}[1]{\left(\begin{matrix}#1\end{matrix}\right)}

\nc{\pimix}{{\pi_{\mathsf{mix}}}}
\nc{\BPC}{{B_{\mathsf{PC}}}}
\nc{\size}{\mathrm{size}}
\nc{\OLIVE}{\texttt{OLIVE}}
\nc{\NP}{\textsf{NP}}
\nc{\RP}{\textsf{RP}}
\nc{\cprp}{c_{\mathsf{PRP}}}
\nc{\Mtoy}{\MM_{\mathsf{toy}}}
\nc{\Brute}{\mathtt{Brute}}

\nc{\nuu}{\nu}

\nc{\bel}[1]{\mathbf{b}({#1})}
\nc{\nbel}[1]{\bar{\mathbf{b}}({#1})}
\nc{\sbel}[2]{\mathbf{b}'_{#1}({#2})}
\nc{\nsbel}[2]{\bar{\mathbf{b}}'_{#1}({#2})}

\nc{\bv}{\mathbf{v}}
\nc{\bone}{\mathbf{1}}
\nc{\bX}{\mathbf{X}}
\nc{\be}{\mathbf{e}}
\nc{\bY}{\mathbf{Y}}
\nc{\bG}{\mathbf{G}}
\nc{\bz}{\mathbf{z}}
\nc{\bw}{\mathbf{w}}
\nc{\bA}{\mathbf{A}}
\nc{\bJ}{\mathbf{J}}
\nc{\bK}{\mathbf{K}}
\nc{\bb}{\mathbf{b}}
\nc{\ba}{\mathbf{a}}
\nc{\bs}{\mathbf{s}}
\nc{\bzero}{\mathbf{0}}
\nc{\bi}{\mathbf{i}}
\nc{\Edistinct}{\ME^{\mathsf{distinct}}}
\nc{\bc}{\mathbf{c}}
\nc{\bC}{\mathbf{C}}
\nc{\BR}{\mathbb R}
\nc{\BA}{\mathbb{A}}
\nc{\SA}{\mathscr{A}}
\nc{\BC}{\mathbb C}
\nc{\bx}{\mathbf{x}}
\nc{\bS}{\mathbf{S}}
\nc{\bM}{\mathbf{M}}
\nc{\bR}{\mathbf{R}}
\nc{\bN}{\mathbf{N}}
\nc{\by}{\mathbf{y}}
\nc{\sy}{y}
\nc{\sx}{x}

\nc{\MO}{\mathcal O}
\nc{\MQ}{\mathcal{Q}}
\nc{\CO}{\mathscr{O}}
\nc{\MU}{\mathcal{U}}
\nc{\ME}{\mathcal{E}}
\nc{\MN}{\mathcal{N}}
\nc{\MK}{\mathcal{K}}
\nc{\MM}{\mathcal{M}}
\nc{\MS}{\mathcal{S}}
\nc{\MT}{\mathcal{T}}
\nc{\BF}{\mathbb F}
\nc{\BQ}{\mathbb Q}
\nc{\MX}{\mathcal{X}}
\nc{\MA}{\mathcal{A}}
\nc{\MD}{\mathcal{D}}
\nc{\MB}{\mathcal{B}}
\nc{\MZ}{\mathcal{Z}}
\nc{\MJ}{\mathcal{J}}
\nc{\MW}{\mathcal{W}}
\nc{\MR}{\mathcal{R}}
\nc{\MY}{\mathcal{Y}}
\nc{\ML}{\mathcal{L}}
\nc{\BZ}{\mathbb Z}
\nc{\BN}{\mathbb N}
\nc{\ep}{\epsilon}
\nc{\gapfn}[1]{\varepsilon_{#1}}
\nc{\ggapfn}[2]{\varphi_{#1}({#2})}
\nc{\epsahk}{\gapfn{0}}
\nc{\BH}{\mathbb H}
\nc{\BG}{\mathbb{G}}
\nc{\D}{\Delta}
\nc{\MF}{\mathcal{F}}
\nc{\One}{\mathbbm{1}}
\nc{\bOne}{\mathbf{1}}
\nc{\Aopt}{\mathcal{A}^{\rm opt}}
\nc{\Amul}{\mathcal{A}^{\rm mul}}

\nc{\SP}{\mathsf P}
\nc{\SQ}{\mathsf Q}

\nc{\DO}{\accentset{\circ}{\D}}
\nc{\mf}{\mathfrak}
\nc{\mfp}{\mathfrak{p}}
\nc{\mfq}{\mf{q}}
\nc{\Sp}{\mbox{Spec}}
\nc{\Spm}{\mbox{Specm}}
\nc{\hookuparrow}{\mathrel{\rotatebox[origin=c]{90}{$\hookrightarrow$}}}
\nc{\hookdownarrow}{\mathrel{\rotatebox[origin=c]{-90}{$\hookrightarrow$}}}
\nc{\hra}{\hookrightarrow}
\nc{\tra}{\twoheadrightarrow}
\nc{\sgn}{{\rm sgn}}
\nc{\aut}{{\rm Aut}}
\nc{\Hom}{{\rm Hom}}
\nc{\img}{{\rm Im}}
\DMO{\id}{Id}
\DMO{\supp}{supp}
\DMO{\KL}{KL}
\nc{\kld}[2]{\KL({#1}||{#2})}
\nc{\ren}[2]{D_2({#1}||{#2})}
\nc{\chisq}[2]{\chi^2({#1}||{#2})}
\nc{\tvd}[2]{D_{\mathsf{TV}}\left({#1}, {#2}\right)}
\nc{\hell}[2]{H^2({#1}, {#2})}
\DMO{\BSS}{BSS}
\DMO{\BES}{BES}
\DMO{\BGS}{BGS}
\DMO{\poly}{poly}
\nc{\indep}{\perp}
\DMO{\sink}{sink}
\DMO{\nosink}{nosink}
\nc{\sinks}{s^{\sink}}
\nc{\sinkobs}{o^{\sink}}
\nc{\fp}[1]{\MP_1({#1})}
\nc{\BO}{\mathbb{O}}
\nc{\BT}{\mathbb{T}}

\nc{\RR}{\mathbb{R}}
\nc{\NN}{\mathbb{N}}
\nc{\Gradient}{\nabla}
\DMO{\diag}{diag}
\nc{\norm}[1]{\left \lVert #1 \right \rVert}
\DMO*{\EE}{\mathbb{E}}
\nc{\LPN}{\mathsf{LPN}}
\DMO{\Ber}{Ber}
\nc{\Regress}{\mathtt{Regress}}
\nc{\LFC}{\mathtt{LearnFromCorr}}
\nc{\RegressAlg}{\mathtt{RegressAlg}}
\nc{\DrawTraj}{\mathtt{DrawTrajectory}}
\nc{\pizero}{{\pi_{\mathsf{zero}}}}
\nc{\Tred}{{T_{\mathsf{red}}}}
\nc{\epred}{{\epsilon_{\mathsf{red}}}}
\nc{\TriAlg}{\mathtt{GenerateTriangleLPN}}
\nc{\Alg}{\mathtt{Alg}}
\nc{\AffSample}{\mathtt{AffSample}}

\nc{\br}{\mathbf{r}}
\nc{\TV}{{\mathsf{TV}}}
\DMO{\Law}{Law}
\DMO{\Sym}{Sym}
\nc{\bu}{\mathbf{u}}
\nc{\Reg}{\mathtt{Reg}}
\nc{\Breg}{B_{\mathsf{Reg}}}

\DMO{\PR}{Pr}
\renewcommand{\Pr}{\PR}
\DMO*{\Prr}{Pr}
\nc{\E}{\mathbb{E}}
\nc{\ra}{\rightarrow}
\renewcommand{\t}{\top}

\title{On Learning Parities with Dependent Noise}

\author{Noah Golowich\thanks{Email: \texttt{nzg@mit.edu}. Supported by a Fannie \& John Hertz Foundation Fellowship and an NSF Graduate Fellowship.} \\ MIT \and Ankur Moitra\thanks{Email: \texttt{moitra@mit.edu}. Supported in part by a Microsoft Trustworthy AI Grant, an ONR grant and a David and Lucile Packard Fellowship.} \\ MIT \and Dhruv Rohatgi\thanks{Email: \texttt{drohatgi@mit.edu}. Supported by a U.S. DoD NDSEG Fellowship.} \\ MIT}

\date{\today}

\begin{document}

\maketitle 

\begin{abstract}
In this expository note we show that the learning parities with noise (LPN) assumption is robust to weak dependencies in the noise distribution of small batches of samples. This provides a partial converse to the linearization technique of \cite{arora2011new}. The material in this note is drawn from a recent work by the authors \cite{golowich2024exploration}, where the robustness guarantee was a key component in a cryptographic separation between reinforcement learning and supervised learning.
\end{abstract}

\section{Introduction}

Learning parities with noise (LPN) is a central problem in average-case complexity, with numerous connections to learning theory and cryptography \cite{blum2003noise, pietrzak2012cryptography}. In the standard search version of the problem, the goal is to recover a fixed $n$-bit secret $\sk \in \BF_2^n$ from independent samples from the distribution $\LPN_{n,\delta}(\sk)$ defined below, for some parameter $\delta \in (0,1/2)$:

\begin{definition}[LPN distribution]
Let $n \in \NN$, $\delta \in (0,1/2)$, and $\sk \in \BF_2^n$. We define the \emph{LPN} distribution $\LPN_{n,\delta}(\sk)$ with secret $\sk$ and noise level $1/2-\delta$ to be the distribution of the random variable $(u,\langle u,\sk\rangle+e)$, where $u \sim \Unif(\BF_2^n)$ and $e \sim \Ber(1/2-\delta)$ are independent.
\end{definition}

Despite intense interest, there is no known polynomial-time algorithm for learning parities with noise even when the noise level is only $\Omega(\log^2(n)/n)$ \cite{brakerski2018anonymous}. In fact, when the noise level is constant, the best known algorithm has time complexity $2^{O(n/\log n)}$ \cite{blum2003noise}, only slightly outperforming exhaustive search. As a result, it is widely conjectured that the problem is computationally intractable. Instantiations of this conjecture have been used to conditionally construct a variety of cryptographic primitives, including symmetric encryption \cite{gilbert2008encrypt,applebaum2009fast}, public-key encryption \cite{alekhnovich2003more, yu2016cryptography}, commitment schemes \cite{jain2012commitments}, and collision-resistant hashing \cite{yu2019collision}. The hardness of LPN has also been employed to prove computational lower bounds throughout theoretical machine learning, from learning Hidden Markov Models \cite{mossel2005learning} and reinforcement learning in block MDPs \cite{golowich2023exploring,golowich2024exploration} to learning quantum states \cite{gollakota2022hardness} and separations between multimodal and unimodal learning \cite{karchmer2024stronger}.

In cryptography, \emph{robustness properties} of concrete computational assumptions such as LPN are of great theoretical and practical importance, due to the possibility of side-channel attacks. As one example, what happens to an encryption scheme where the secret key is partially leaked? Robustness of the related Learning with Errors (LWE) assumption to non-uniform distributions over the secret \cite{goldwasser2010robustness, brakerski2020hardness} implies that schemes based on LWE are still secure. As another example, what happens if the encryptions rely on imperfect randomness? Little is known about the robustness of LWE or LPN to non-standard noise distributions; see \cref{sec:related} for related prior work.

In this note, we give a self-contained exposition of an LPN robustness property recently proven by the authors as part of a larger result \cite{golowich2024exploration}. Essentially, this property asserts that learning parities with \emph{small batches of weakly dependent noise} is still hard, under the standard LPN hardness assumption with appropriate parameters. A formal statement requires the following definition:

\begin{definition}[Batch LPN distribution]
Let $n,k \in \NN$, $p \in \Delta(\BF_2^k)$, and $\sk \in \BF_2^n$. We define the \emph{batch LPN} distribution $\LPN_{n,p}(\sk)$ with secret $\sk$ and joint noise distribution $p$ to be
\[(u^i,\langle u^i,\sk\rangle + e^i)_{i=1}^k \sim \LPN_{n,p}(\sk)\]
where $u^1,\dots,u^k \sim \Unif(\BF_2^n)$ and $(e^1,\dots,e^k) \sim p$ are independent.
\end{definition}

Given independent batch samples from $\LPN_{n,p}(\sk)$, the goal of batch LPN with joint noise distribution $p$ is to recover $\sk$. Note that $k=1$ and $p = \Ber(1/2-\delta)$ exactly corresponds to one sample from $\LPN_{n,\delta}(\sk)$, i.e. the standard LPN distribution with noise level $1/2-\delta$. However, for $k \geq 2$, if $p$ is not a product distribution, the connection to standard LPN is not evident. 

\paragraph{When is batch LPN hard?} It's easy to construct seemingly innocuous joint distributions $p$ for which batch LPN can actually be solved in $\poly(n)$ time \--- e.g. the uniform distribution over vectors in $\BF_2^k$ with Hamming weight exactly one \--- because the dependencies enable simulating noiseless samples. This algorithmic approach is generalized by the Arora-Ge \emph{linearization attack} \cite{arora2011new}, which remains the state-of-the-art for batch LPN. However, this attack still only applies to joint distributions with specific structure (see \cref{sec:related}). In the absence of this structure, is batch LPN computationally hard? In analogy with the seminal robustness results discussed above, can this be proven under a standard LPN assumption? 

\Cref{lemma:construct-corr-lpn} takes a first step in this direction. The family of \emph{Santha-Vazirani} sources \cite{santha1986generating} is a broad family of distributions that naturally models imperfect, physical noise sources where the bias of each new bit may have some dependence on previous bits, but the dependence is never too strong:

\begin{definition}[c.f. \cite{santha1986generating}]\label{def:sv-source}
Let $k\in\NN$, $\delta \in (0,1/2)$, and $p \in \Delta(\BF_2^k)$. We say $p$ is a \emph{$\delta$-Santha-Vazirani source} if for all $i \in [k]$ and $x \in \BF_2^k$ it holds that
\[\Prr_{X\sim p}[X_i=1|X_{<i}=x_{<i}] \in [1/2-\delta, 1/2+\delta].\]
\end{definition}

In words, a distribution over $\BF_2^k$ is a $\delta$-Santha-Vazirani source if each bit is biased at most $\delta$ from uniform, under any conditioning of all preceding bits. \Cref{lemma:construct-corr-lpn} shows that a sample from the batch LPN distribution with a $\delta$-Santha-Vazirani source can be efficiently constructed using $k$ independent samples from the standard LPN distribution with noise level $1/2 - O_k(\delta)$.

\begin{theorem}[\cite{golowich2024exploration}]\label{lemma:construct-corr-lpn}
There is an algorithm $\EntLPN$ with the following property. Fix $n,k \in \NN$, $\delta\in(0,1/2^{k+3})$, and $p \in \Delta(\BF_2^k)$. Suppose that $p$ is a $\delta$-Santha-Vazirani source. For every $\sk \in \BF_2^n$, for independent samples $(a_i,y_i)_{i=1}^k$ from $\LPN_{n,2^{k+2}\delta}(\sk)$, the output of $\EntLPN((a_i,y_i)_{i=1}^k, p, \delta)$ has distribution $\LPN_{n,p}(\sk)$. Moreover, the time complexity of $\EntLPN$ is $\poly(n, 2^k)$.
\end{theorem}

As an immediate consequence, for any constant $k$ and $\delta$-Santha-Vazirani source $p \in \Delta(\BF_2^k)$, the standard problem of learning $\sk$ given independent samples from $\LPN_{n,2^{k+2}\delta}(\sk)$ is polynomial-time reducible to the problem of learning $\sk$ given independent batches from $\LPN_{n,p}(\sk)$.

\subsection{Related work}\label{sec:related}

Batch LPN was first studied by Arora and Ge, who called it ``learning parities with structured noise'' \cite{arora2011new}. They required that there is some $w \in \BF_2^k$ such that $\Pr_{e,e' \sim p}[e+e' = w] = 0$. Under this structural assumption on the support of $p$, they showed that $\sk$ can be learned using $n^{O(k)}$ time and batches, which is polynomial when $k$ is constant.\footnote{In fact, their algorithm works in a more general model where the noise vector of each batch is chosen adversarially from the support of $p$. It also applies to larger fields (i.e. the learning with errors problem).} 

Little subsequent progress has been made on understanding batch LPN, either via algorithms or hardness. Most relevant is a result due to \cite{bartusek2019new}, which shows hardness of batch LPN with $k=2$ and a specific structured choice of $p$, under a standard LPN assumption, so long as the number of samples is not too large. Their choice of $p$ is not a $\delta$-Santha-Vazirani source, so their result is incomparable to \cref{lemma:construct-corr-lpn}. 

Other forms of robustness for LPN and LWE have received significant attention, but much remains unknown. As previously mentioned, the LWE assumption is known to be robust to non-uniform, high-entropy distributions for the secret \cite{brakerski2020hardness}. While LPN with high-entropy secret distributions is believed to be hard, no formal reduction from standard LPN is known \cite{yao2016robustness}. The LWE assumption is known to be robust to uniform noise distributions rather than discrete Gaussian when the sample complexity is bounded; establishing a reduction in greater generality is a long-standing open problem motivated by the computational burden of sampling a discrete Gaussian \cite{micciancio2013hardness,dottling2013lossy}. Finally, LPN is known to be robust to noise levels depending on the inner product \cite{bellizia2021learning}, another setting motivated by physical noise models.

\paragraph{Outline.} The next two sections of this note are drawn (almost verbatim) from \cite[Section 8.2]{golowich2024exploration}, and are included so that this note is self-contained. In \cref{sec:overview} below, we give an overview of the proof of \cref{lemma:construct-corr-lpn}, and in \cref{sec:proof} we give the formal proof. 

In \cref{sec:discussion} we discuss future directions and make some concluding remarks.

\section{Proof overview}\label{sec:overview}

An obvious attempt at proving \cref{lemma:construct-corr-lpn} would be to show that there is some distribution $\til p \in \Delta(\BF_2^k)$ so that if $X \sim \Ber(1/2 - 2^{k+2}\delta)^{\otimes k}$ and $Y \sim \til p$ are independent, then $Z = X+Y$ is distributed according to $p$. If this were true, one could construct a batch distributed according to $\LPN_{n,p}(\sk)$ by sampling $Y\sim\til p$ and adding $Y$ to the responses in the input batch $(a_i,y_i)_{i=1}^k$. 

Unfortunately, this is actually false, even for $k=2$. Suppose that $p$ is the distribution of the random vector $Z = (e_1 + b, e_2 + b) \in \BF_2^2$ where $e_1, e_2 \sim \Ber(1/2-\sqrt{\delta})$ and $b \sim \Ber(1/2)$ are independent. Then $Z_2|Z_1=0$ and $Z_2|Z_1=1$ are both $O(\delta)$-close to $\Ber(1/2)$, so $p$ is an $O(\delta)$-Santha-Vazirani source. However, since $\TV(Z_1+Z_2, \Ber(1/2)) = \Omega(\delta)$, the data processing inequality implies that if we have $Z=X+Y$ for independent random variables $X,Y$, then 
\begin{align*}
\TV(X_1+X_2, \Ber(1/2)) 
&\geq \TV(X_1+Y_1+X_2+Y_2, \Ber(1/2)+Y_1+Y_2) \\ 
&= \TV(Z_1+Z_2, \Ber(1/2)) \\ 
&= \Omega(\delta).
\end{align*}
Thus, if the input noise distribution is $X \sim \Ber(1/2-\eta)^{\otimes 2}$, then (by e.g. \cref{lem:bernoulli-convolve}) the bias $\eta$ must be $\Omega(\sqrt{\delta})$. Since we want the input noise distribution to have bias $O(\delta)$, this approach does not work.

For this particular choice of $p$, there is a simple transformation to construct the desired batch, with joint noise distribution $p$, from a single sample from $\LPN_{n,2\delta}(\sk)$, but it's not immediately clear how to generalize that transformation even slightly \--- e.g., if $p$ is instead the distribution of $Z = (e_1+b, e_2+b, e_3+b)$ for $e_1,e_2,e_3 \sim \Ber(1/2-\delta)$ and $b \sim \Ber(1/2)$. 

The takeaway of the above analysis seems to be that we need to add noise that is \emph{correlated} with the noise in the input LPN samples. Indeed, suppose that we have managed to construct a batch of $k-1$ LPN samples $(a'_i,y'_i)_{i=1}^{k-1}$ with joint noise distribution $p_{1:k-1}$, so all that remains is to construct a random variable $(a'_k,y'_k)$ where $a'_k \sim \Unif(\BF_2^n)$ and $y'_k$ has the appropriate conditional distribution: \[y'_k - \langle a'_k, \sk\rangle \sim \Ber(\Pr_{Z\sim p}[Z_k=1|Z_j = y'_j-\langle a'_j,\sk\rangle \, \forall j \in [k-1]]).\]
By assumption on $p$, this Bernoulli random variable has some bias $\delta' \in [-\delta,\delta]$. Thus, if we knew $\delta'$, then we could construct $(a'_k,y'_k)$ from a fresh input sample $(a_k,y_k) \sim \LPN_{n,2^{k+2}\delta}(\sk)$ by adding $y_k$ to an independent Bernoulli random variable with bias $\delta'/(2^{k+3}\delta)$. This gets around the above obstacle because $\delta'$ implicitly depends on the noise terms $(y_i-\langle a_i,\sk\rangle)_{i=1}^{k-1}$ of the input samples. Unfortunately, it therefore also depends on $\sk$, so we cannot hope to learn it.

\paragraph{Solution: adding affine noise (in $\sk$).} The first insight (previously employed for key-dependent message cryptography \cite{applebaum2009fast}) is to observe that we can sometimes ``simulate'' adding something to the noise term of an LPN sample $(a_k,y_k)$ that depends on $\sk$, by appropriately adjusting not just $y_k$ but also $a_k$. In particular, suppose that $(a_k, y_k) \sim \LPN_{n,\delta}(\sk)$ and $F(\sk) = \alpha + \langle \beta,\sk\rangle$ where $\alpha,\beta$ are independent of $(a_k,y_k)$. Then the random variable
\[(a'_k,y'_k) := (a_k - \beta, y_k + \alpha)\]
satisfies that $a'_k \sim \Unif(\BF_2^n)$ and moreover that, conditioned on $a'_k$,
\[y'_k - \langle a'_k, \sk\rangle = y_k + \alpha - \langle a_k - \beta, \sk\rangle = y_k - \langle a_k,\sk\rangle + F(\sk).\]
Thus, for all intents and purposes, we have ``added'' the affine function $F(\sk)$ to the noise term. It remains to argue that we can write the random variable $\Ber(1/2 - \delta'/(2^{k+3}\delta))$ as a (random) affine function in $\sk$.

\paragraph{Linearization.} Note that $1/2 - \delta'/(2^{k+3}\delta)$ is some known function (say, $q$) of the unknown vector $z := (y'_j-\langle a'_j,\sk\rangle)_{j=1}^{k-1}$. Thinking of $z$ as an affine function in $\sk$, we can write down the coefficients of any affine function that is a linear combination of $z_1,\dots,z_{k-1}$ (along with a constant term), i.e. $F_0 + F_1z_1 + \dots + F_{k-1}z_{k-1}$. Thus, it suffices to show that there is a distribution $\mu$ over such linear combinations so that
\begin{equation} F_0 + F_1 z_1 + \dots + F_{k-1} z_{k-1} \sim \Ber(q(z))\label{eq:fz-q-constraint}\end{equation}
for all fixed $z$. We accomplish this by a perturbation argument. It's easy to see that if $q$ is the constant function $q(z) = 1/2$, then we can simply define $\mu := \Ber(1/2)^{\otimes k}$. Of course $q$ may not be constant, but by assumption we know that $q$ is entry-wise \emph{close} to constant. Moreover, for each $z$ the induced constraint \cref{eq:fz-q-constraint} on $\mu$ is linear in the density function of $\mu$, and we can show that the system of constraints (across all $z$) is essentially well-conditioned. Hence, a small perturbation to $q$ preserves satisfiability of the system. Some care has to be taken to ensure that e.g. the perturbed $\mu$ is still a distribution, but this is the main idea.

\section{Proof of \cref{lemma:construct-corr-lpn}}\label{sec:proof}

We proceed to the formal proof of \cref{lemma:construct-corr-lpn}, starting with the following least singular bound that will be needed for the perturbation argument.

\begin{lemma}
  \label{lem:b-conditioning}
Fix $k \in \NN$ and let $B \in \RR^{2^k-1 \times 2^k-1}$ be the matrix with rows and columns indexed by nonzero vectors in $\{0,1\}^k$, where $B_{uv} = \langle u,v\rangle \bmod{2}$ for any nonzero vectors $u,v \in \{0,1\}^k$. Then we have the following properties (note that we are treating $B$ as a real-valued matrix, not a matrix over $\BF_2$):
\begin{itemize}
    \item $B\mathbbm{1} = 2^{k-1}\mathbbm{1}$.
    \item The least singular value of $B$ satisfies $\sigma_{\min}(B) \geq 2^{(k-2)/2}$.
\end{itemize}
\end{lemma}

\begin{proof}
For any nonzero vector $v \in \{0,1\}^k$, the number of $u \in \{0,1\}^k$ such that $\langle u,v\rangle \equiv 1 \bmod{2}$ is exactly $2^{k-1}$. The first property immediately follows. Next, for any distinct nonzero vectors $u,v \in \{0,1\}^k$, we observe that the number of $w \in \{0,1\}^k$ such that $\langle u,w\rangle \equiv \langle v,w \rangle \equiv 1 \bmod{2}$ is exactly $2^{k-2}$. Thus, $B^\t B = 2^{k-2} \mathbbm{1}\mathbbm{1}^\t + 2^{k-2} I$. It follows that $\lambda_{\min}(B^\t B) \geq 2^{k-2}$ and thus $\sigma_{\min}(B) \geq 2^{(k-2)/2}$.
\end{proof}

Using \cref{lem:b-conditioning}, we can now formally prove that \cref{eq:fz-q-constraint} is solvable for any near-uniform function $q$. Note that the distance to uniform needs to be exponentially small in the batch size $k$, hence the factor of $2^{k+2}$ in the bias of the input LPN samples in \cref{lemma:construct-corr-lpn}. It is an interesting question whether this dependence can be removed.

\begin{lemma}\label{lemma:cond-dist-linearization}
Fix $k \in \NN$, and pick any function $q: \BF_2^k \to [1/2 - 2^{-(k+3)}, 1/2+2^{-(k+3)}]$. Then there is a random variable $F = (F_0,\dots,F_k)$ on $\BF_2^{k+1}$ such that for all $z \in \BF_2^k$, it holds that
\[\Pr(F_0 + z_1 F_1 + \dots + z_k F_k \equiv 1 \bmod{2}) = q(z).\]
Moreover, there is an algorithm that takes $q$ as input and samples $F$ in time $2^{O(k)}$.
\end{lemma}

\begin{proof}
We represent $q$ as an element of $\RR^{\BF_2^k}$, and we also identify $\Delta(\BF_2^{k+1})$ with the simplex in $\RR^{\BF_2^{k+1}}$. Throughout this proof, we let $\mathbbm{1}$ denote the all-ones real vector, of appropriate dimension per the context. Set $\delta := \norm{q - \frac{1}{2}\mathbbm{1}}_\infty \leq 2^{-(k+3)}$. Let $A \in \RR^{2^{k+1} \times 2^k}$ be the matrix with rows indexed by $\BF_2^{k+1}$ and columns indexed by $\BF_2^k$, so that for any $f \in \BF_2^{k+1}$ and $z \in \BF_2^k$, we define $A_{fz} := \mathbbm{1}[f_0+f_1z_1+\dots+f_kz_k \equiv 1 \bmod{2}]$. Then for any distribution $\mu \in \Delta(\BF_2^{k+1})$, it is clear that for all $z \in \BF_2^k$,
\begin{equation}
\Prr_{F\sim \mu}(F_0+z_1F_1+\dots+z_kF_k \equiv 1 \bmod{2}) = (A^\t \mu)_z.\label{eq:fz-pr-linear}\end{equation}
We construct a distribution $\mu^\st \in \Delta(\BF_2^{k+1})$ satisfying $A^\t \mu^\st = q$ as follows. First, define $\bar \mu \in \Delta(\BF_2^{k+1})$ by 
\[\bar\mu := \Ber(q(\mathbf{0})) \times \Ber(1/2)^{\otimes k}\]
where $\mathbf{0}$ denotes the all-0s vector (taken to be of appropriate dimension per the context).  
Next, define $\MZ = \BF_2^k \setminus \mathbf{0}$ and $\MF = \{0\} \times \MZ \subseteq \BF_2^{k+1}$, and define the matrix $B := (A_{\MF\MZ})^\t$. Note that $B$ is exactly the matrix defined in \cref{lem:b-conditioning}. Let $q_\MZ \in \BR^\MZ$ denote the vector whose entry corresponding to $z \in \MZ$ is $q(z)$. We construct $\mu^\st$, viewed as a vector in $\BR^{\BF_2^{k+1}}$, by separately defining its components: $\mu^\st_\MF$ for entries in $\MF$, $\mu^\st_{\mathbf{0}}$ for the $\mathbf{0}$-entry, and all other entries: 
\begin{align} \mu^\st_\MF &:=  \bar\mu_\MF + B^{-1} \left(q_\MZ - \frac{1}{2}\mathbbm{1}\right)\label{eq:pdfmu1} \\
\mu^\st_{\mathbf{0}} &:=  \bar\mu_{\mathbf{0}} + \mathbbm{1}^\t B^{-1} \left(\frac{1}{2}\mathbbm{1} - q_\MZ\right)\label{eq:pdfmu2}\\
  \mu^\st_f &:=  \bar\mu_f \qquad \forall f \in \BF_2^{k+1} \setminus (\MF \cup \{\mathbf{0}\}).\label{eq:pdfmu3}
\end{align}
We must first show that $\mu^\st$ is a well-defined distribution on $\BF_2^{k+1}$. First, by \cref{lem:b-conditioning} and the fact that $B$ is a symmetric matrix, we have that $B$ is invertible. Moreover, 
\begin{align}\label{eq:bdiff}\norm{B^{-1}\left(q_\MZ - \frac{1}{2}\mathbbm{1}\right)}_\infty \leq \norm{B^{-1}\left(q_\MZ - \frac{1}{2}\mathbbm{1}\right)}_2 \leq 2^{(2-k)/2} \norm{q_\MZ - \frac{1}{2}\mathbbm{1}}_2 \leq 2\norm{q_\MZ - \frac{1}{2}\mathbbm{1}}_\infty \leq 2\delta\end{align}
where the second inequality is by \cref{lem:b-conditioning} and the final inequality is by definition of $\delta$. Also, by construction, we know that $\frac{\min(q(\mathbf{0}), 1-q(\mathbf{0}))}{2^k}\leq \bar\mu_f \leq \frac{1}{2^k}$ for all $f \in \BF_2^{k+1}$. It follows that for all $f \in \MF$,
\[0 \leq \frac{1/2-\delta}{2^k} - 2\delta \leq \frac{\min(q(\mathbf{0}), 1-q(\mathbf{0}))}{2^k} - 2\delta \leq \mu^\st_f \leq \frac{1}{2^k} + 2\delta \leq 1,\]
where the first inequality is because $(2^{k+2}+2)\delta \leq 1$, the second inequality is by definition of $\delta$, the third and fourth inequalities use \cref{eq:pdfmu1} and \cref{eq:bdiff}, and the final inequality is because $\delta \leq 1/4$. Next, since $B\mathbbm{1} = 2^{k-1} \mathbbm{1}$ (\cref{lem:b-conditioning}), we have $B^{-1}\mathbbm{1} = 2^{1-k}\mathbbm{1}$, so \[|\mu^\st_{\mathbf{0}}-\bar\mu_{\mathbf{0}}| = 2^{1-k}\left|\frac{\langle\mathbbm{1},\mathbbm{1}\rangle}{2} - \langle q_\MZ,\mathbbm{1}\rangle\right| \leq 2 \norm{\frac{1}{2} \mathbbm{1} - q_\MZ}_\infty \leq 2\delta\] and thus, like above, $0 \leq \mu^\st_{\mathbf{0}} \leq 1$. Finally, since $\bar\mu$ is a distribution, it's immediate from \cref{eq:pdfmu3} that $0 \leq \mu^\st_f \leq 1$ for all $f \in \BF_2^{k+1} \setminus (\MF \cup \{\mathbf{0}\})$. This shows that all entries of $\mu^\st$ are between $0$ and $1$. Next, observe that
\[\mathbbm{1}^\t \mu^\st = \mathbbm{1}^\t \bar\mu + \mathbbm{1}^\t B^{-1}\left(q_\MZ - \frac{1}{2}\mathbbm{1}\right) + \mathbbm{1}^\t B^{-1} \left(\frac{1}{2}\mathbbm{1} - q_\MZ\right) = 1\] 
where the final equality uses that $\bar\mu$ is a distribution. We conclude that $\mu^\st$ is a distribution.

Next, we must prove that $A^\t \mu^\st = q$. First recall that $\bar\mu$ is the distribution $\Ber(q(\mathbf{0})) \times \Ber(1/2)^{\otimes k}$. Thus, one can see from \cref{eq:fz-pr-linear} that $(A^\t \bar\mu)_{\mathbf{0}} = q(\mathbf{0})$, whereas for any $z \in \MZ$, $(A^\t \bar\mu)_z = 1/2$. Now since $A_{f,\mathbf{0}} = 0$ for all $f \in \MF \cup \{\mathbf{0}\}$, and $\mu^\st_{f} = \bar\mu_{f}$ for all $f \in \BF_2^{k+1}\setminus(\MF\cup\{\mathbf{0}\})$, we see that $(A^\t \mu^\st)_{\mathbf{0}} = (A^\t \bar\mu)_{\mathbf{0}} = q(\mathbf{0})$ as desired. Moreover, using that $A_{\mathbf{0},z} = 0$ for all $z \in \BF_2^k$,
\[(A^\t \mu^\st)_\MZ = (A^\t \bar\mu)_\MZ + A_{\MF\MZ}^\t (\mu^\st_\MF - \bar\mu_\MF) = \frac{1}{2}\mathbbm{1} + BB^{-1}\left(q_\MZ - \frac{1}{2}\mathbbm{1}\right) = q_\MZ\]
as needed. Since $\MZ\cup\{\mathbf{0}\} = \BF_2^k$, this proves that $A^\t\mu^\st = q$.

Finally, the claimed time complexity bound for sampling $F \sim \mu^\st$ follows from the fact that the probability mass function $\mu^\st$ (of size $2^{k+1}$) has an explicit formula (described in \cref{eq:pdfmu1,eq:pdfmu2,eq:pdfmu3}) and can be computed in time $2^{O(k)}$.
\end{proof}



We can now complete the proof of \cref{lemma:construct-corr-lpn}, following the sketch described above for generating the $k$-th sample of the batch given the first $k-1$ samples (and the target distribution $p$).

\begin{namedproof}{\cref{lemma:construct-corr-lpn}}
Throughout, $\sk \in \BF_2^n$ is fixed and unknown to the algorithm. The algorithm $\EntLPN$ constructs $(a'_i,y'_i)_{i=1}^k$ iteratively from $i=1$ to $k$. At each step $i$, we will maintain the invariant that $(a'_1,\dots,a'_i,(y'_1-\langle a'_1,\sk\rangle,\dots,y'_i-\langle a'_i,\sk\rangle))$ is distributed according to $\Unif(\BF_2^n)^{\otimes i} \times p_{1:i}$, for any $\sk \in \BF_2^n$, where $p_{1:i}$ is the marginal distribution of $Z_{1:i}$ for $Z \sim p$.

At step $i=1$, $\EntLPN$ sets
\[a'_1 := a_1\]
\[e'_1 \sim \Ber\left(\frac{1}{2} - \frac{1}{2}\frac{\frac{1}{2} - p_1(1)}{2^{k+2}\delta}\right)\]
\[y'_1 := y_1 + e'_1\]
where $p_1$ is the marginal distribution of $Z_1$ for $Z\sim p$. Note that $e'_1$ is well-defined since $p_1(1) \in [1/2-\delta,1/2+\delta]$ (by the fact that $p$ is a $\delta$-Santha-Vazirani source). By construction, $a'_1 \sim \Unif(\BF_2^n)$. Moreover, since $y_1 = \langle a_1,\sk\rangle + e_1$, where $e_1 \sim \Ber(1/2 - 2^{k+2} \delta)$ is independent of $a_1$, it follows that $y'_1 - \langle a'_1,\sk\rangle = e_1 + e'_1$ is independent of $a'_1$ and (by \cref{lem:bernoulli-convolve}) has distribution $\Ber(p_1(1))$. Thus, indeed $(a'_1, y'_1 - \langle a'_1,\sk\rangle) \sim \Unif(\BF_2^n) \times p_1$.

Now fix $1 < i \leq k$ and suppose that $\EntLPN$ has constructed $(a'_j,y'_j)_{j=1}^{i-1}$ with the desired distribution, using $(a_j,y_j)_{j=1}^{i-1}$. At step $i$, $\EntLPN$ computes the function $p^{(i)}: \BF_2^{i-1} \to [0,1]$ defined by
\begin{align}
  \label{eq:pi-defn}
  p^{(i)}(z_{1:i-1}) := \frac{1}{2} - \frac{\frac{1}{2} - \Pr_{Z\sim p}[Z_i=1|Z_{1:i-1}=z_{1:i-1}]}{2^{k+3}\delta}.
\end{align}
Since $p$ is a $\delta$-Santha-Vazirani source, it holds that $\Pr_{Z\sim p}[Z_i=1|Z_{1:i-1}=z_{1:i-1}] \in [1/2-\delta,1/2+\delta]$, and thus $p^{(i)}(z_{1:i-1}) \in [1/2 - 1/2^{k+3}, 1/2 + 1/2^{k+3}]$, for all $z_{1:i-1} \in \BF_2^{i-1}$. Applying the algorithm guaranteed by \cref{lemma:cond-dist-linearization}, $\EntLPN$ samples a random variable $(F^{(i)}_0,F^{(i)}_1,\dots,F^{(i)}_{i-1})$ with the property that
\begin{align}
  \label{eq:Fproperty}
  \Pr(F^{(i)}_0+F^{(i)}_1z_1+\dots+F^{(i)}_{i-1}z_{i-1} \equiv 1 \bmod{2}) = p^{(i)}(z_{1:i-1})
\end{align}
for all $z_{1:i-1} \in \BF_2^{i-1}$. 
Finally, $\EntLPN$ sets
\[a'_i := a_i + F^{(i)}_1 a'_1 + \dots + F^{(i)}_{i-1} a'_{i-1}\]
\[y'_i := y_i + F^{(i)}_0 + F^{(i)}_1 y'_1 + \dots + F^{(i)}_{i-1} y'_{i-1}.\] 
It remains to argue that $(a'_j,y'_j)_{j=1}^i$ has the desired distribution. To see this, first observe that the following four random variables are, by construction, mutually independent:
\begin{itemize}\setlength\itemsep{0.1em}
    \item $(a'_j,y'_j)_{j=1}^{i-1}$,
    \item $a_i$,
    \item $y_i - \langle a_i, \sk\rangle$,
    \item $(F_0^{(i)},\dots,F_{i-1}^{(i)})$.
\end{itemize}
Thus, suppose we condition on the first $i-1$ samples $(a'_j,y'_j)_{j=1}^{i-1}$. Then still $a_i \sim \Unif(\BF_2^n)$, and moreover $a_i$ is independent of $(F^{(i)}_0,\dots,F^{(i)}_{i-1})$. Thus, conditioned on $(a_j', y_j')_{j=1}^{i-1}$, we have that $a'_i$ is distributed according to $\Unif(\BF_2^n)$.

Next we also condition on $a'_i$. Since $a'_i$ is independent of $(F^{(i)}_0,\dots,F^{(i)}_{i-1})$ conditioned on $(a_j', y_j')_{j=1}^{i-1}$, the distribution of $(F^{(i)}_0,\dots,F^{(i)}_{i-1})$ conditioned on $\{ (a_j', y_j')_{j=1}^{i-1}, a_i' \}$ still satisfies the property \cref{eq:Fproperty}: for every $z \in \BF_2^{i-1}$, the distribution of $F^{(i)}_0+F^{(i)}_1z_1+\dots+F^{(i)}_{i-1}z_{i-1}$ is still $\Ber(p^{(i)}(z_{1:i-1}))$. We have that
\begin{align}
y'_i - \langle a'_i, \sk\rangle 
&= y_i - \langle a_i, \sk\rangle + F^{(i)}_0 + F^{(i)}_1(y'_1 - \langle a'_1, \sk\rangle) + \dots + F^{(i)}_{i-1}(y'_{i-1} - \langle a'_{i-1}, \sk\rangle).\label{eq:yprime-diff}
\end{align}
Using the property \cref{eq:Fproperty}, we have $F^{(i)}_0+\sum_{j=1}^{i-1} F^{(i)}_j (y'_j - \langle a'_j, \sk\rangle) \sim \Ber(p^{(i)}((y'_j-\langle a'_j,\sk\rangle)_{j=1}^{i-1}))$. Also $y_i - \langle a_i,\sk\rangle \sim \Ber(1/2 - 2^{k+2} \delta)$ is independent of $(F^{(i)}_j)_{j=0}^{i-1}$. 
Thus, combining \cref{eq:yprime-diff}, the definition of $p^{(i)}$ in \cref{eq:pi-defn}, and \cref{lem:bernoulli-convolve}, we get
\[y'_i - \langle a'_i, \sk\rangle \sim \Ber\left(\Pr_{Z\sim p}\left[Z_i=1|Z_j = y'_j-\langle a'_j,\sk\rangle \quad \forall j \in [i-1]\right]\right).\]
It follows that the conditional distribution of $(a'_i, y'_i-\langle a'_i,\sk\rangle)$, given $(a'_j,y'_j)_{j=1}^{i-1}$, is \[\Unif(\BF_2^n) \times \Ber\left(\Pr_{Z\sim p}\left[Z_i=1|Z_j = y'_j-\langle a'_j,\sk\rangle \quad \forall j \in [i-1]\right]\right).\]
By the inductive hypothesis, we get that the distribution of $(a'_j,y'_j-\langle a'_j,\sk\rangle)_{j=1}^i$ is 
\[\Unif(\BF_2^n)^{\otimes i} \times p_{1:i}\] as desired. Finally, we note that the claimed time complexity bound on $\EntLPN$ follows from the algorithm description and \cref{lemma:cond-dist-linearization}.
\end{namedproof}

In the above proof, we used the following basic fact:

\begin{lemma}
  \label{lem:bernoulli-convolve}
Let $\delta_1, \delta_2 \in [-1/2,1/2]$. If $Z_1 \sim \Ber(1/2 -\delta_1)$ and $Z_2 \sim \Ber(1/2 - \delta_2)$ are independent, and $Z_3 := Z_1 + Z_2 \mod{2}$, then $Z_3 \sim \Ber(1/2 - 2\delta_1\delta_2)$.
\end{lemma}
\begin{proof}
We can check that
\[\Pr[Z_3=1] = \left(\frac{1}{2}-
\delta_1\right)\left(\frac{1}{2}+\delta_2\right)+\left(\frac{1}{2}+\delta_1\right)\left(\frac{1}{2}- \delta_2\right) = \frac{1}{2}-2\delta_1\delta_2\]
as needed.
\end{proof}

\section{Discussion}\label{sec:discussion}

\cref{lemma:construct-corr-lpn} shows that the LPN hardness assumption is robust to weak dependencies in the noise terms of small batches of samples. To the authors' knowledge, this result has not been previously observed. The authors applied it to prove a computational separation between reinforcement learning and supervised learning \cite{golowich2024exploration}, where it was crucial due to the dependent nature of data in reinforcement learning. We hope it may have further applications in learning theory and cryptography.

There are a number of ways in which \cref{lemma:cond-dist-linearization} could potentially be strengthened. First, the current result shows that batch LPN with a $\delta$-Santha-Vazirani source is as hard as LPN with noise level $1/2 - 2^{k+2} \delta$. This is only non-vacuous for fairly high noise levels, since it requires $\delta \leq 2^{-(k+3)}$. Intuitively, one might expect that batch LPN with a $\delta$-Santha-Vazirani source is as hard as LPN with noise level $1/2-\delta$. Can this be formally proven? Second, the reduction incurs time complexity exponential in the batch size $k$, which is currently unavoidable since the reduction takes as input an explicit description of the joint noise distribution $p \in \Delta(\BF_2^k)$. Is there a natural and general model where $p$ is succinctly described and larger batch sizes can be handled? For instance, consider the variant of LPN with a sampling oracle. At the $m$-th sampling query, the oracle takes as input a circuit $C: \BF_2^{m-1} \to [1/2-\delta,1/2+\delta]$ from the learning agent, and returns an LPN sample with noise term $e_m \sim \Ber(C(e_{1:m-1}))$, where $e_1,\dots,e_{m-1}$ are the previous noise terms. Is this problem as hard as standard LPN with noise level $1/2-\delta$?

\paragraph{Acknowledgments.} We would like to thank Vinod Vaikuntanathan for encouraging us to post this note and providing additional context about known robustness properties of LWE/LPN.

\bibliographystyle{amsalpha}
\bibliography{bib}

\end{document}